\documentclass[pra,twocolumn,groupedaddress,superscriptaddress,amsmath,amssymb,a4paper]{revtex4}
\usepackage{geometry}
\usepackage{graphicx}
\usepackage{mathrsfs}
\usepackage{amssymb}
\usepackage{amsmath}
\usepackage{amsthm}
\usepackage{amsbsy}
\usepackage{bm}
\usepackage{hyperref}
\usepackage[T1]{fontenc}

\newtheorem{proposition}{Proposition}

\newtheorem{corollary}{Corollary}

\theoremstyle{definition}

\newtheorem{example}{Example}

\newcommand{\bra}[1]{\langle #1|}
\newcommand{\ket}[1]{| #1 \rangle }

\begin{document}

\title{Ultimate entanglement robustness of two-qubit states against general local noises}

\author{Sergey N. Filippov}

\affiliation{Moscow Institute of Physics and Technology,
Institutskii Per. 9, Dolgoprudny, Moscow Region 141700, Russia}

\affiliation{Institute of Physics and Technology, Russian Academy
of Sciences, Nakhimovskii Pr. 34, Moscow 117218, Russia}

\author{Vladimir V. Frizen}

\affiliation{Moscow Institute of Physics and Technology,
Institutskii Per. 9, Dolgoprudny, Moscow Region 141700, Russia}

\author{Daria V. Kolobova}

\affiliation{Moscow Institute of Physics and Technology,
Institutskii Per. 9, Dolgoprudny, Moscow Region 141700, Russia}

\begin{abstract}
We study the problem of optimal preparation of a bipartite
entangled state, which remains entangled the longest time under
action of local qubit noises. We show that for unital noises such
a state is always maximally entangled, whereas for nonunital
noises, it is not. We develop a decomposition technique relating
nonunital and unital qubit channels, based on which we find the
explicit form of the ultimately robust state for general local
noises. We illustrate our findings by amplitude damping processes
at finite temperature, for which the ultimately robust state
remains entangled up to two times longer than conventional
maximally entangled states.
\end{abstract}

\maketitle

\section{Introduction}
\label{section:introduction}

Quantum communication is one of the most developed subfields in
the practical realization of quantum information
protocols~\cite{gisin-thew-2007}. Dense
coding~\cite{bennett-1992}, quantum
teleportation~\cite{bennett-1993}, remote state
preparation~\cite{bennett-2005}, and some cryptographic
schemes~\cite{ekert-1991,deutsch-1996,acin-2007} are based on the
phenomenon of entanglement. Entanglement is also widely used in
other quantum information applications~\cite{horodecki-2009}. When
two laboratories $A$ and $B$ are taken into account, by entangled
state we understand a density operator $\varrho^{AB}$ (unit trace
positive-semidefinite operator acting on some Hilbert space
$\mathcal{H}$), which does not belong to a closure of separable
states of the form $\varrho^{AB} = \sum_k p_k \varrho_k^A \otimes
\varrho_k^B$, $p_k \geqslant 0$, $\sum_k p_k =
1$~\cite{werner-1989}. Entangled states cannot be created by local
operations and classical communication from factorized
states~\cite{nielsen-1999}, so entanglement between
non-interacting laboratories $A$ and $B$ can only be created via
sending parts of a locally prepared initial entangled state
$\varrho_{\rm in}^{AB}$ to $A$ and $B$, respectively (transmission
of an entangled state can also be a stage in a more involved
process such as entanglement swapping~\cite{zukowski-1993}). Since
$A$ and $B$ are supposed to be far apart, transmission of the
entangled state is carried out by means of local quantum channels
$\Phi^{AB} = \Phi_1^A \otimes \Phi_2^B$; see Fig.~\ref{figure1}.
Quantum channel $\Phi:\mathcal{B}(\mathcal{H}) \mapsto
\mathcal{B}(\mathcal{H})$ is a completely positive
trace-preserving map that describes the result of quantum system
transformation due to unavoidable interaction with environment
(quantum
noise)~\cite{nielsen-2000,breuer-petruccione-2002,holevo-2012}.
The longer the quantum channels between the entanglement source
and laboratories $A$, $B$, the noisier and less entangled becomes
the output state, $\varrho_{\rm out}^{AB} = (\Phi_1^A \otimes
\Phi_2^B)[\varrho_{\rm
in}^{AB}]$~\cite{yu-2004,bellomo-2008,yu-eberly-2009,wang-2013,shaham-2015,aolita-2015}.
The length of the quantum channels can be included in the above
description by time $t$ quantifying the duration of the
system-environment interaction: $\varrho^{AB}(t) = \Phi_1^A(t)
\otimes \Phi_2^B(t) [\varrho_{\rm in}^{AB}]$, with $\Phi_1^A(0)$
and $\Phi_2^B(0)$ being identity transformations (${\rm Id}$).
Preservation of entanglement of the state $\varrho^{AB}(t)$ is the
primary goal for implementing entanglement-based protocols. In
fact, if $A$ and $B$ are both qubit systems and $\varrho^{AB}(t)$
is entangled, then by sending the same state $\varrho_{\rm
in}^{AB}$ through a quantum channel $\Phi_1^A(t) \otimes
\Phi_2^B(t)$ many times, one can distill maximally entangled
states $\varrho_+ = \ket{\psi_+}\bra{\psi_+}$, $\ket{\psi_+} =
\frac{1}{\sqrt{2}} (\ket{00}+\ket{11})$ that are useful in
entanglement-based applications~\cite{horodecki-prl-1997}. Given
quantum noises $\Phi_1^A(t)$ and $\Phi_2^B(t)$, the entanglement
lifetime of the state $\varrho_{\rm in}^{AB}$ is defined as the
minimal time $\tau$ such that $\varrho^{AB}(t)$ is separable for
all $t \geqslant \tau$. In other words, the entanglement lifetime
(also referred to as disentangling time) is the time of
entanglement sudden death~\cite{yu-2004}. The maximal possible
entanglement lifetime $\widetilde{\tau} = \max_{\varrho_{\rm
in}^{AB}} \tau$ provides the fundamental restriction on the length
of quantum channels to $A$ and $B$. The state
$\widetilde{\varrho}_{\rm in}^{AB}$, which maximizes entanglement
lifetime, exhibits the ultimate entanglement robustness to local
noises $\Phi_1^A(t) \otimes \Phi_2^B(t)$. If
$\widetilde{\varrho}_{\rm in}^{AB}$ is the most robust to the loss
of entanglement with respect to the dynamical map $\Phi_1^A(t)
\otimes \Phi_2^B(t)$, then separability of $\Phi_1^A(t) \otimes
\Phi_2^B(t)[\widetilde{\varrho}_{\rm in}^{AB}]$ implies
separability of $\Phi_1^A(t) \otimes \Phi_2^B(t)[\varrho_{\rm
in}^{AB}]$ for all input states $\varrho_{\rm in}^{AB}$. Note that
the output of the channel $\Phi_1^A(\widetilde{\tau}) \otimes
\Phi_2^B(\widetilde{\tau})$ is separable for all possible input
states, i.e., such a channel is entanglement
annihilating~\cite{moravcikova-ziman-2010,filippov-rybar-ziman-2012,filippov-ziman-2013,filippov-ziman-2-2013,filippov-2014,filippov-2015}.

\begin{figure}[b]
\includegraphics[width=8.5cm]{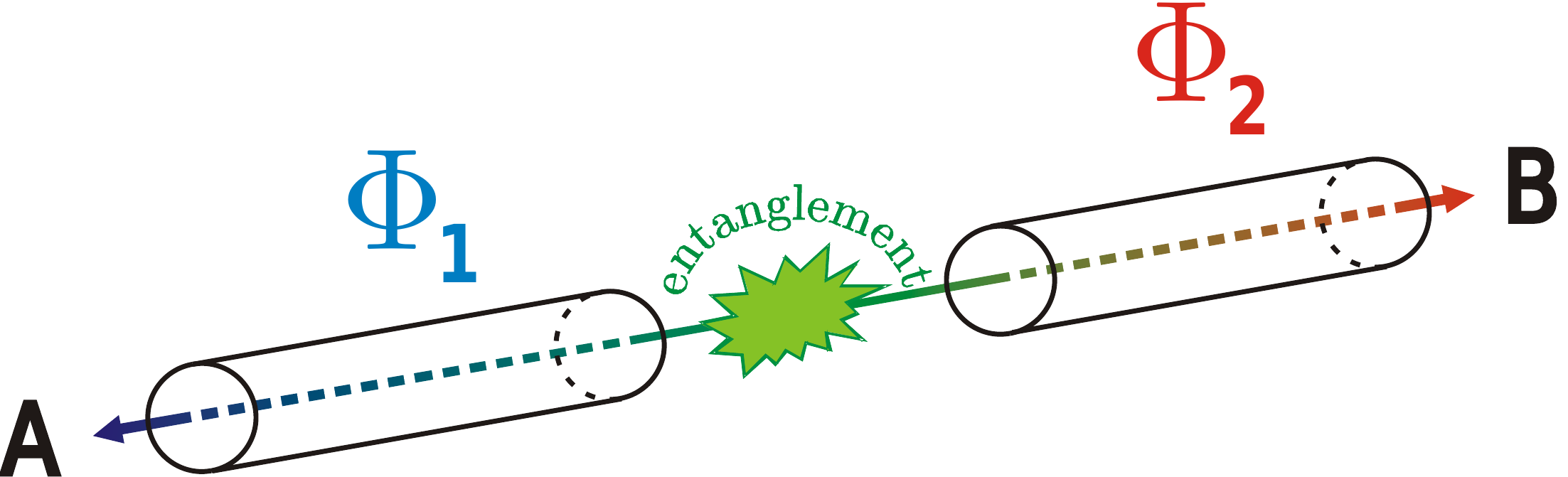}
\caption{\label{figure1} Transmission of entangled state through
local quantum channels.}
\end{figure}

Despite the fact that entanglement of a two-qubit system can be
readily and precisely verified via the Peres-Horodecki
criterion~\cite{peres-1996,horodecki-1996} or
concurrence~\cite{hill-wootters-1997,wootters-1998}, it is not
that easy to resolve the maximin problem of entanglement lifetime
$\widetilde{\tau}$ even for a simple semigroup dynamics
$\Phi_1^A(t) \otimes \Phi_2^B(t) = e^{\mathcal{L}_1^A t} \otimes
e^{\mathcal{L}_2^B t}$ describing generalized amplitude damping
processes~\cite{jakobczyk-2004,ali-2009}. It is also not known how
to find the optimal state $\widetilde{\varrho}_{\rm in}^{AB}$
analytically. There are three distinguished exceptions, however.
The first one is the case of one-sided noiseless evolution, when
$\Phi_1^A(t) \equiv {\rm Id}$, i.e. one part of the entangled
system is perfectly preserved; then the maximally entangled state
$\varrho_+$ has ultimate
robustness~\cite{horodecki-2003,konrad-2008,xu-2009}. The second
exception is the case of local depolarizing noises, with
$\varrho_+$ being ultimately robust~\cite{tiersch-2009}. The third
exception is the case of local unital~\footnote{A linear map
$\Upsilon: \mathcal{B}(\mathcal{H}) \mapsto
\mathcal{B}(\mathcal{H})$ is called unital if it preserves
identity operator, i.e. $\Upsilon[I] = I$.} two-qubit dynamical
maps $\Upsilon(t) \otimes \Upsilon(t)$, for which the maximally
entangled state $\varrho_+$ is the most robust to the loss of
entanglement too~\cite{filippov-rybar-ziman-2012}. In this paper,
we extend these results to the case of general local unital
channels $\Upsilon_1^A(t) \otimes \Upsilon_2^B(t)$ and prove that
the maximally entangled state $\varrho_+$ is optimal for the
transmission of entanglement through such channels. It is tempting
to conclude that the maximally entangled state $\varrho_+$
exhibits ultimate robustness to general local two-qubit noises
$\Phi_1^A(t) \otimes \Phi_2^B(t)$; however, this is not
true~\cite{wang-2015,ziman-buzek-2005,ziman-buzek-2007} and we
show that explicitly in this paper. Moreover, we analytically find
the initial two-qubit state $\widetilde{\varrho}$, which is the
most robust to a given nonunital local two-qubit dynamical map
$\Phi_1(t) \otimes \Phi_2(t)$. The use of the optimal initial
state for entanglement distribution enables essential extension of
the length of communication lines, which we demonstrate by
examples of generalized amplitude damping processes.

The paper is organized as follows. In Sec.~\ref{section:unital},
we consider two-qubit local unital dynamical maps $\Upsilon(t)
\otimes \Upsilon'(t)$ and prove that the ultimately robust state
is necessarily maximally entangled. We also find a criterion to
check if the map $\Upsilon(t) \otimes \Upsilon'(t)$ is
entanglement annihilating, based on which one can
straightforwardly calculate the maximal entanglement lifetime. In
Sec.~\ref{section:nonunital-robustness}, we show how the results
for unital dynamical maps are related with those for nonunital
ones, provided a special decomposition is known. We find the
explicit form of such a decomposition of nonunital channels in
Sec.~\ref{section:nonunital-decomposition}. In
Sec.~\ref{section:nonunital-gad}, we apply the developed theory to
nonunital channels describing the process of amplitude damping due
to qubit interaction with the environment of finite temperature.
In Sec.~\ref{section:conclusions}, brief conclusions are given.

\section{Unital channels}
\label{section:unital}

A unital qubit channel $\Upsilon$ is necessarily random
unitary~\cite{landau-streater-1993} and with a suitable choice of
input and output bases, can be represented in the
form~\cite{ruskai-2002}
\begin{equation}
\label{Upsilon} \Upsilon[X] = \frac{1}{2} {\rm tr}[X] I +
\frac{1}{2} \sum_{i=1}^3 \lambda_i {\rm tr}[\sigma_i X] \sigma_i,
\end{equation}

\noindent where $\sigma_1,\sigma_2,\sigma_3$ is a conventional set
of Pauli operators such that $\sigma_3\ket{0} = \ket{0}$ and
$\sigma_3\ket{1} = - \ket{1}$. The map \eqref{Upsilon} is known to
be positive if $-1 \leqslant \lambda_1,\lambda_2,\lambda_3
\leqslant 1$, completely positive if $1 \pm \lambda_3 \geqslant |
\lambda_1 \pm \lambda_2 |$, and entanglement breaking if
$|\lambda_1| + |\lambda_2| + |\lambda_3| \leqslant
1$~\cite{ruskai-2003}. We will associate every map $\Upsilon$ with
the corresponding vector $\boldsymbol{\lambda} =
(\lambda_1,\lambda_2,\lambda_3)^{\top}$.

Matrix representation $M_{ij}(\Upsilon) = \frac{1}{2}{\rm tr}
\big[ \sigma_i \Upsilon[\sigma_j] \big]$, $i,j = 0, \ldots,3$,
$\sigma_0 = I$, of the map~\eqref{Upsilon} reads
\begin{equation}
M(\Upsilon) = \left(%
\begin{array}{cccc}
  1 & 0 & 0 & 0 \\
  0 & \lambda_1 & 0 & 0 \\
  0 & 0 & \lambda_2 & 0 \\
  0 & 0 & 0 & \lambda_3 \\
\end{array}%
\right) = {\rm diag}(1,\boldsymbol{\lambda}^{\top}).
\end{equation}

A local unital two-qubit map $\Upsilon \otimes \Upsilon$ composed
of identical unital maps $\Upsilon$ is known to be entanglement
annihilating if $\lambda_1^2 + \lambda_2^2 + \lambda_3^2 \leqslant
1$~\cite{filippov-rybar-ziman-2012}, with the maximally entangled
state $\varrho_+$ having the longest entanglement lifetime. Some
sufficient and (separately) necessary conditions for entanglement
annihilation of the general local unital two-qubit map $\Upsilon
\otimes \Upsilon'$ are listed in
Ref.~\cite{filippov-rybar-ziman-2012}. We fill the gap in analysis
of such maps and provide a criterion of entanglement annihilation.

\begin{proposition}
\label{proposition:unital} Suppose $\Upsilon$ and $\Upsilon'$ are
positive qubit maps. Then the map $\Upsilon \otimes \Upsilon'$ is
positive and entanglement annihilating if and only if
$|\lambda_i|,|\lambda_i'| \leqslant 1$, $i=1,2,3$ and
$\boldsymbol{\lambda} P \boldsymbol{\lambda}' \leqslant 1$ for all
signed permutation matrices $P$.
\end{proposition}

\begin{proof}
\textit{Sufficiency}. Due to a convex structure of separable
states, a map $\Upsilon \otimes \Upsilon'$ is entanglement
annihilating if and only if $\Upsilon \otimes
\Upsilon'[\ket{\psi}\bra{\psi}]$ is separable for all pure states
$\ket{\psi}$. On the other hand, any pure two-qubit state
$\ket{\psi}$ can be represented as a linear combination of
Bell-like states $\ket{\varphi_i} = \sigma_i \otimes I
\ket{\psi_+}$, $i=0,\ldots,3$:
\begin{equation}
\ket{\psi} = \sum_{i=0}^3 c_i \ket{\varphi_i} = C \otimes I
\ket{\psi_+},
\end{equation}

\noindent where $C = \sum_{i=0}^3 c_i \sigma_i$. Denote $\Phi_C[X]
= C X C^{\dag}$; then the density operator of any two-qubit pure
state takes the form
\begin{equation}
\label{any-pure-state-operator} \ket{\psi}\bra{\psi} = \Phi_C
\otimes {\rm Id} [\ket{\psi_+}\bra{\psi_+}].
\end{equation}

Kraus representation of the map $\Upsilon'$ is well
known~\cite{ruskai-2002} and reads $\Upsilon'[X] = \sum_{j=0}^3
q_j' \sigma_j X \sigma_j$, where real parameters $\{q_j'\}$ are
uniquely expressed through parameters $\{\lambda_j'\}$. Since $I
\otimes \sigma_j' \ket{\psi_+} = (\sigma_j')^{\top} \otimes I
\ket{\psi_+}$, we get
\begin{eqnarray}
\label{Upsilon-one-sided} {\rm Id} \otimes \Upsilon'
[\ket{\psi_+}\bra{\psi_+}] & = & \sum_{j=0}^3 q_j'
(\sigma_j')^{\top} \otimes I \ket{\psi_+}\bra{\psi_+}
(\sigma_j')^{\top} \otimes I \nonumber\\
& = & \Upsilon' \otimes {\rm Id} [\ket{\psi_+}\bra{\psi_+}],
\end{eqnarray}

\noindent where we have taken into account that
$(\sigma_2')^{\top}=-\sigma_2'$ and $(\sigma_j')^{\top}=\sigma_j'$
if $j=0,1,3$. Combining \eqref{any-pure-state-operator} and
\eqref{Upsilon-one-sided}, we can express the action of the map
$\Upsilon \otimes \Upsilon'$ on any pure state as follows:
\begin{eqnarray}
\Upsilon \otimes \Upsilon' [\ket{\psi}\bra{\psi}] & = & (\Upsilon
\otimes \Upsilon') \circ (\Phi_C \otimes {\rm Id})
[\ket{\psi_+}\bra{\psi_+}] \nonumber\\
 & = & (\Upsilon \circ \Phi_C
\otimes {\rm Id}) \circ ({\rm Id} \otimes \Upsilon') [\ket{\psi_+}\bra{\psi_+}] \nonumber\\
& = & \Upsilon \circ \Phi_C \circ \Upsilon' \otimes {\rm Id}
[\ket{\psi_+}\bra{\psi_+}]. \label{output}
\end{eqnarray}

Therefore, the map $\Upsilon \otimes \Upsilon'$ is entanglement
annihilating if and only if the output state \eqref{output} is
separable for all matrices $C$. The necessary and sufficient
criterion of separability of two-qubit states provides the
reduction criterion~\cite{horodecki-1999}, which states that the
two-qubit state $\varrho$ is separable if and only if ${\cal R}
\otimes {\rm Id}[\varrho] \geqslant 0$, where the action of qubit
map ${\cal R}$ reads ${\cal R}[X] = {\rm tr}[X] I - X$. Thus, the
state \eqref{output} is separable if and only if ${\cal R} \circ
\Upsilon \circ \Phi_C \circ \Upsilon' \otimes {\rm Id}
[\ket{\psi_+}\bra{\psi_+}] \geqslant 0$, or equivalently
\begin{equation}
\bra{\chi} \left( {\cal R} \circ \Upsilon \circ \Phi_C \circ
\Upsilon' \otimes {\rm Id} [\ket{\psi_+}\bra{\psi_+}] \right)
\ket{\chi} \geqslant 0
\end{equation}

\noindent for all two-qubit states $\ket{\chi}$. Similarly to
Eq.~\eqref{any-pure-state-operator}, we represent $\ket{\chi} = D
\otimes I \ket{\psi_+}$ and conclude that $\Upsilon \otimes
\Upsilon'$ is entanglement annihilating if and only if
\begin{equation}
\label{CD} \bra{\psi_+} \left( \Phi_{D^{\dag}} \circ {\cal R}
\circ \Upsilon \circ \Phi_C \circ \Upsilon' \otimes {\rm Id}
[\ket{\psi_+}\bra{\psi_+}] \right) \ket{\psi_+} \geqslant 0
\end{equation}

\noindent for all matrices $C$ and $D$. Recalling $\ket{\psi_+} =
\frac{1}{\sqrt{2}}\sum_{k=0}^1 \ket{k} \otimes \ket{k}$,
Eq.~\eqref{CD} is equivalent to
\begin{equation}
\label{DRC} \sum_{k,l=0}^1 \bra{k}  \Big( \Phi_{D^{\dag}} \circ
{\cal R} \circ \Upsilon \circ \Phi_C \circ \Upsilon'
[\ket{k}\bra{l}] \Big)  \ket{l} \geqslant 0.
\end{equation}

The basis of matrix units $E_{kl} = \ket{k}\bra{l}$ is orthonormal
in the sense of Hilbert-Schmidt inner product $(X,Y)={\rm
tr}[X^{\dag} Y]$. So is the basis of operators $\{
\frac{1}{\sqrt{2}}\sigma_j \}_{j=0}^3$, and hence $E_{kl} =
\sum_{j=0}^3 W_{kl,j} \frac{1}{\sqrt{2}}\sigma_j$ and
$\sum_{k,l=0}^1 W_{kl,i}^{\ast} W_{kl,j} = \delta_{ij}$.
Eq.~\eqref{DRC} takes the form
\begin{eqnarray}
0 & \leqslant & \sum_{k,l=0}^1 {\rm tr} \left[ E_{kl}^{\dag}
\Phi_{D^{\dag}} \circ {\cal R} \circ \Upsilon \circ \Phi_C \circ
\Upsilon'
[E_{kl}] \right] \nonumber\\
& = & \frac{1}{2} \sum_{i,j=0}^3 \sum_{k,l=0}^1 W_{kl,i}^{\ast}
W_{kl,j} {\rm tr} \left[ \sigma_{i}^{\dag} \Phi_{D^{\dag}} \circ
{\cal R} \circ
\Upsilon \circ \Phi_C \circ \Upsilon' [\sigma_j] \right] \nonumber\\
& = & \frac{1}{2} \sum_{i=0}^3 {\rm tr} \left[ \sigma_{i}
\Phi_{D^{\dag}} \circ {\cal R} \circ \Upsilon \circ \Phi_C \circ
\Upsilon' [\sigma_i] \right] \nonumber\\
& = & {\rm tr} \left[ M (\Phi_{D^{\dag}} \circ {\cal R} \circ
\Upsilon
\circ \Phi_C \circ \Upsilon') \right] \nonumber\\
& = &  {\rm tr} \left[ M (\Phi_{D^{\dag}}) M({\cal R})
M(\Upsilon) M(\Phi_C) M(\Upsilon') \right] \nonumber\\
& = & {\rm tr} \left[ M (\Phi_{D^{\dag}}) {\rm
diag}(1,-\boldsymbol{\lambda}^{\top}) M(\Phi_C) {\rm
diag}(1,{\boldsymbol{\lambda}'}^{\top}) \right] \nonumber\\
& = & (1,-\boldsymbol{\lambda}^{\top}) \, M
(\Phi_{D^{\dag}})^{\top} \!\! \ast M(\Phi_C)  \, \left(%
\begin{array}{c}
  1 \\
  \boldsymbol{\lambda}'
\end{array}%
\right) \nonumber\\
& = & (1,-\boldsymbol{\lambda}^{\top}) \, M
(\Phi_C) \ast M(\Phi_D) \, \left(%
\begin{array}{c}
  1 \\
  \boldsymbol{\lambda}'
\end{array}%
\right), \label{main-inequality}
\end{eqnarray}

\noindent where $\ast$ denotes the Hadamard pointwise product,
i.e. $(M \ast N)_{ij} = M_{ij}N_{ij}$.

To know matrix representations of maps $\Phi_C$ and $\Phi_D$, we
use singular-value decompositions,
\begin{eqnarray}
C &=& U_C \left(%
\begin{array}{cc}
  \sqrt{1+\sin\alpha_{C}} & 0 \\
  0 & \sqrt{1-\sin\alpha_{C}} \\
\end{array}%
\right) V_C, \\
D &=& U_D \left(%
\begin{array}{cc}
  \sqrt{1+\sin\alpha_{D}} & 0 \\
  0 & \sqrt{1-\sin\alpha_{D}} \\
\end{array}%
\right) V_D,
\end{eqnarray}

\noindent which explicitly take into account that the states
$\ket{\psi}$ and $\ket{\chi}$ are normalized, i.e. ${\rm
tr}[C^{\dag}C] = {\rm tr}[D^{\dag}D] = 2$. Here, $U_C$, $V_C$,
$U_D$, $V_D$ are unitary operators and $0 \leqslant \alpha_C,
\alpha_D \leqslant \frac{\pi}{2}$. Therefore,
\begin{eqnarray}
&& M(\Phi_C) = \left(%
\begin{array}{c|c}
  1 & {\bf 0}^{\top} \\
  \hline
  {\bf 0} & Q_{U_C} \\
\end{array}%
\right) \left(%
\begin{array}{c|c}
  1 & {\bf t}_C^{\top} \\
  \hline
  {\bf t}_C & T_C \\
\end{array}%
\right) \left(%
\begin{array}{c|c}
  1 & {\bf 0}^{\top} \\
  \hline
  {\bf 0} & Q_{V_C} \\
\end{array}%
\right), \qquad \label{MC} \\
&& M(\Phi_D) = \left(%
\begin{array}{c|c}
  1 & {\bf 0}^{\top} \\
  \hline
  {\bf 0} & Q_{U_D} \\
\end{array}%
\right) \left(%
\begin{array}{c|c}
  1 & {\bf t}_D^{\top} \\
  \hline
  {\bf t}_D & T_D \\
\end{array}%
\right) \left(%
\begin{array}{c|c}
  1 & {\bf 0}^{\top} \\
  \hline
  {\bf 0} & Q_{V_D} \\
\end{array}%
\right), \qquad \label{MD}
\end{eqnarray}

\noindent where $Q_U$ is a $3 \times 3$ orthogonal matrix
corresponding to channel $\Phi_U$, ${\bf 0}=(0,0,0)^{\top}$, ${\bf
t}_{C(D)} = (0,0,\sin\alpha_{C(D)})^{\top}$, and $T_{C(D)} = {\rm
diag}\left( \cos\alpha_{C(D)}, \cos\alpha_{C(D)}, 1 \right)$.

By ${\bf u}_{C(D)1}, {\bf u}_{C(D)2}, {\bf u}_{C(D)3}$, denote
three orthonormal columns of the matrix $Q_{U_{C(D)}}$ and, by
${\bf v}_{C(D)1}^{\top}, {\bf v}_{C(D)2}^{\top}, {\bf
v}_{C(D)3}^{\top}$, denote three orthonormal rows of the matrix
$Q_{V_{C(D)}}$. Introduce the vectors
\begin{equation}
{\bf u}_{kl} = {\bf u}_{Ck} \ast {\bf u}_{Dl}, \quad {\bf v}_{kl}
= {\bf v}_{Ck} \ast {\bf v}_{Dl}.
\end{equation}

\noindent Then the direct calculation of the Hadamard product $M
(\Phi_C) \ast M(\Phi_D)$ yields
\begin{equation}
M (\Phi_C) \ast M(\Phi_D) = \left(%
\begin{array}{c|c}
  1 &  \sin\alpha_C \sin\alpha_D {\bf v}_{33}^{\top} \\
  \hline
  \sin\alpha_C \sin\alpha_D {\bf u}_{33} & S \\
\end{array}%
\right),
\end{equation}

\noindent where
\begin{eqnarray}
&& S = {\bf u}_{33} {\bf v}_{33}^{\top} + \cos\alpha_C [ {\bf
u}_{13} {\bf v}_{13}^{\top} + {\bf u}_{23} {\bf v}_{23}^{\top} ] \nonumber\\
&& + \cos\alpha_D [ {\bf u}_{31} {\bf v}_{31}^{\top} + {\bf
u}_{32} {\bf
v}_{32}^{\top} ] \nonumber\\
&& + \cos\alpha_C \cos\alpha_D [ {\bf u}_{11} {\bf v}_{11}^{\top}
+ {\bf u}_{12} \ast {\bf v}_{12}^{\top} + {\bf u}_{21} {\bf
v}_{21}^{\top} + {\bf u}_{22} {\bf v}_{22}^{\top} ]. \nonumber\\
\end{eqnarray}

By the Cauchy-Bunyakovsky-Schwarz inequality $|( u_{kl})_x| + |(
u_{kl})_y| + |( u_{kl})_z| \leqslant |{\bf u}_{Ck}| \cdot |{\bf
u}_{Dl}| = 1$ and $|( v_{kl})_x| + |( v_{kl})_y| + |( v_{kl})_z|
\leqslant |{\bf v}_{Ck}| \cdot |{\bf v}_{Dl}| = 1$. Thus, all the
vectors ${\bf u}_{kl}$ and ${\bf v}_{kl}$ belong to the
\textit{octahedron} with vertices $(\pm 1,0,0)$, $(0, \pm 1, 0)$,
and $(0,0, \pm 1)$. Moreover, since vectors ${\bf u}_{C(D)1}, {\bf
u}_{C(D)2}, {\bf u}_{C(D)3}$ are mutually orthogonal, vectors
${\bf u}_{kl}$ and ${\bf u}_{k'l}$ (${\bf u}_{kl'}$) cannot belong
to the same octant or opposite octants if $k \neq k'$ ($l \neq
l'$). Since vectors ${\bf u}_{kl}$ and ${\bf v}_{kl}$ linearly
contribute to the expression
\begin{eqnarray}
&& \!\!\!\!\!\!\!\!\!\! (1,-\boldsymbol{\lambda}^{\top}) \, M
(\Phi_C) \ast M(\Phi_D) \, \left(%
\begin{array}{c}
  1 \\
  \boldsymbol{\lambda}'
\end{array}%
\right) \nonumber\\
&& \!\!\!\!\!\!\!\!\!\! = 1 - \sin\alpha_C \sin\alpha_D
\boldsymbol{\lambda}^{\top}{\bf u}_{33} + \sin\alpha_C
\sin\alpha_D {\bf v}_{33}^{\top} \boldsymbol{\lambda}' -
\boldsymbol{\lambda}^{\top} S \boldsymbol{\lambda}', \nonumber\\
\label{expression-through-vectors}
\end{eqnarray}

\noindent the minimal value of \eqref{expression-through-vectors}
is achieved if some vectors ${\bf u}_{kl}$ and ${\bf v}_{kl}$
correspond to the extreme points of the octahedron, i.e., to
vectors $(\pm 1,0,0)$, $(0, \pm 1, 0)$, and $(0,0, \pm 1)$.
Without loss of generality it can be assumed that ${\bf u}_{33} =
{\bf v}_{33} = (0,0,1)$, which implies ${\bf u}_{k3} = {\bf
u}_{3k} = {\bf 0}$ and ${\bf v}_{k3} = {\bf v}_{3k} = {\bf 0}$,
$k=1,2$. Then either $S = {\rm diag} (\pm \cos\alpha_C
\cos\alpha_D, \pm
\cos\alpha_C \cos\alpha_D, 1)$ or $S = \left(%
\begin{array}{ccc}
  0 & \pm \cos\alpha_C \cos\alpha_D & 0 \\
  \pm \cos\alpha_C \cos\alpha_D & 0 & 0 \\
  0 & 0 & 1 \\
\end{array}%
\right)$, where signs $\pm$ are not correlated. Inequality
\eqref{main-inequality} reduces to
\begin{eqnarray}
\label{reduced-inequality}
 0 & \leqslant & 1 - \sin\alpha_C \sin\alpha_D (\lambda_3 -
\lambda_3')
- \lambda_3 \lambda_3' \nonumber\\
&& - \cos\alpha_C \cos\alpha_D \left\{
\begin{array}{c}
  \pm \lambda_1 \lambda_1' \pm \lambda_2 \lambda_2' , \\
  \pm \lambda_1 \lambda_2' \pm \lambda_2 \lambda_1' , \\
\end{array} \right.
\end{eqnarray}

\noindent which is fulfilled for all $0 \leqslant
\alpha_C,\alpha_D \leqslant \frac{\pi}{2}$ if
$\boldsymbol{\lambda} \left(%
\begin{array}{ccc}
  \pm 1 & 0 & 0 \\
  0 & \pm 1 & 0 \\
  0 & 0 & 1 \\
\end{array}%
\right) \boldsymbol{\lambda}' \leqslant 1$, $\boldsymbol{\lambda} \left(%
\begin{array}{ccc}
  0 & \pm 1 & 0 \\
  \pm 1 & 0 & 0 \\
  0 & 0 & 1 \\
\end{array}%
\right) \boldsymbol{\lambda}' \leqslant 1$, and
$|\lambda_i|,|\lambda_i'| \leqslant 1$, $i=1,2,3$.

It can easily be checked numerically that in the general case of
arbitrary vectors ${\bf u}_{kl}$ and ${\bf v}_{kl}$,
inequality~\eqref{expression-through-vectors} is fulfilled
whenever $|\lambda_i|,|\lambda_i'| \leqslant 1$, $i=1,2,3$, and
$\boldsymbol{\lambda} P \boldsymbol{\lambda}' \leqslant 1$ for all
signed permutation matrices $P$.

\textit{Necessity}. Let the input state $\ket{\psi} =
\ket{\psi_+}$, then the output state $\Upsilon \otimes \Upsilon'
[\ket{\psi_+}\bra{\psi_+}]$ is separable by Peres-Horodecki
criterion if and only if $1 + \lambda_3 \lambda_3' \pm (\lambda_1
\lambda_1' - \lambda_2 \lambda_2') \geqslant 0$ and $1 - \lambda_3
\lambda_3' \pm (\lambda_1 \lambda_1' + \lambda_2 \lambda_2')
\geqslant 0$. Also, the state ${\cal R} \circ \Upsilon \otimes
\Upsilon' [\ket{\psi_+}\bra{\psi_+}]$ must be separable, which
corresponds to the change $\lambda_i \rightarrow - \lambda_i$. By
permuting indices $(1,2,3)$ of the second qubit, we obtain that
the condition $\boldsymbol{\lambda} P \boldsymbol{\lambda}'
\leqslant 1$ must be fulfilled for all signed permutation matrices
$P$. Permutation of indices corresponds to the change of input
state to the form $\frac{1}{\sqrt{2}}\left( \ket{\varphi} \otimes
\ket{\chi} + \ket{\varphi_{\perp}} \otimes \ket{\chi_{\perp}}
\right)$, where $\{\ket{\varphi},\ket{\varphi_{\perp}}\}$ and
$\{\ket{\chi},\ket{\chi_{\perp}}\}$ are bases of eigenvectors of
some Pauli operators.
\end{proof}

\begin{corollary}
\label{corollary} Suppose $1 \geqslant \lambda_1 \geqslant
\lambda_2 \geqslant \lambda_3 \geqslant 0$ and $1 \geqslant
\lambda_1' \geqslant \lambda_2' \geqslant \lambda_3' \geqslant 0$;
then the local two-qubit unital map $\Upsilon \otimes \Upsilon'$
is entanglement annihilating if and only if
$\boldsymbol{\lambda}^{\top} \boldsymbol{\lambda}' = \lambda_1
\lambda_1' + \lambda_2 \lambda_2' + \lambda_3 \lambda_3' \leqslant
1$.
\end{corollary}

\begin{proof}
It is not hard to see that $\boldsymbol{\lambda}^{\top} P
\boldsymbol{\lambda}'$ achieves maximum among signed permutation
matrices $P$ if $P = I$. Then the statement of
Corollary~\ref{corollary} follows directly from
Proposition~\ref{proposition:unital}.
\end{proof}

In the necessity part of Proposition~\ref{proposition:unital}, we
have noticed that ultimate robust states to local noises
$\Upsilon(t) \otimes \Upsilon'(t)$ are the states of the form
$\frac{1}{\sqrt{2}}\left( \ket{\varphi} \otimes \ket{\chi} +
\ket{\varphi_{\perp}} \otimes \ket{\chi_{\perp}} \right)$, where
$\{\ket{\varphi},\ket{\varphi_{\perp}}\}$ and
$\{\ket{\chi},\ket{\chi_{\perp}}\}$ are bases of eigenvectors of
some Pauli operators.

\begin{proposition}
\label{proposition:ultimate-unital} Suppose a local two-qubit
unital noise $\Upsilon(t) \otimes \Upsilon'(t)$, with matrix
representations of $\Upsilon(t)$, $\Upsilon'(t)$ being diagonal in
the basis of Pauli operators $\sigma_1$, $\sigma_2$, $\sigma_3$.
Then the state with ultimate entanglement robustness is the
maximally entangled state $\ket{\psi_{\Upsilon\otimes\Upsilon'}} =
\frac{1}{\sqrt{2}}\left( \ket{\varphi} \otimes \ket{\chi} +
\ket{\varphi_{\perp}} \otimes \ket{\chi_{\perp}} \right)$, where
$\{\ket{\varphi},\ket{\varphi_{\perp}}\}$ and
$\{\ket{\chi},\ket{\chi_{\perp}}\}$ are orthogonal eigenvectors of
some Pauli operators ($\sigma_1$, or $\sigma_2$, or $\sigma_3$).
\end{proposition}

\begin{example}
\label{example-1} Consider an amplitude damping process of a
two-level system (see, e.g., Ref.~\cite{nielsen-2000}, section
8.3.5), when the temperature of the environment is so high
(thermal energy $kT \gg \Delta E$, energy level separation) that
the rate of spontaneous emission equals the rate of spontaneous
absorbtion. If this is the case, then Markov approximation leads
to the following master equation in the interaction
picture~(\cite{breuer-petruccione-2002}, section 10.1):
\begin{eqnarray}
\frac{d\varrho}{dt} &=& \gamma \left( \sigma_{+} \varrho
\sigma_{-} - \frac{1}{2} \left\{ \sigma_{-}\sigma_{+} , \varrho
\right\} \right) \nonumber\\
&& + \gamma \left( \sigma_{-} \varrho \sigma_{+} - \frac{1}{2}
\left\{ \sigma_{+}\sigma_{-} , \varrho \right\} \right),
\end{eqnarray}

\noindent where $\{ \cdot, \cdot \}$ denotes anticommutator,
$\sigma_{\pm} = \frac{1}{2}(\sigma_1 \pm \sigma_2)$, and $\gamma
> 0$ is the damping rate. Solution of this master equation results
in a unital map~\eqref{Upsilon} with $\lambda_1(t) = \lambda_2(t)
= e^{-\gamma t}$ and $\lambda_3(t) = e^{-2\gamma t}$.

Suppose two qubits each experiencing amplitude damping in a
high-temperature environment with damping rates $\gamma$ and
$\gamma'$, respectively. Then the maximally entangled state with
one excitation $\ket{\psi} = \frac{1}{\sqrt{2}}(\ket{0} \otimes
\ket{1} + \ket{1} \otimes \ket{0})$ exhibits the maximal
entanglement robustness and the entanglement lifetime is
determined by the equation $\lambda_1(t) \lambda_1'(t) +
\lambda_2(t) \lambda_2'(t) + \lambda_3(t) \lambda_3'(t) = 1$,
i.e., $2 e^{-(\gamma + \gamma')t} + e^{-2 (\gamma + \gamma')t} =
1$. The maximal entanglement lifetime equals $\widetilde{\tau} =
\frac{\ln(\sqrt{2}+1)}{\gamma+\gamma'} \approx
\frac{0.88}{\gamma+\gamma'}$.
\end{example}

\begin{example}
\label{example-2} Suppose a pair of entangled qubits is prepared
in laboratory $A$; one qubit is kept in the quantum memory cell of
laboratory $A$ and the other is sent to laboratory $B$. The qubit
in laboratory $A$ is subjected to amplitude damping in a
high-temperature environment with damping rates $\gamma$, and the
itinerant qubit experiences depolarization with dissipator
$\mathcal{L}=\gamma' \sum_{j=1}^3 \left( \sigma_j \varrho \sigma_j
- \varrho \right)$. Then $\lambda_1(t) = \lambda_2(t) = e^{-\gamma
t}$, $\lambda_3(t) = e^{-2\gamma t}$, and $\lambda_1'(t) =
\lambda_2'(t) = \lambda_3'(t) = e^{-\gamma' t}$. From
Corollary~\ref{corollary}, it follows that the state $\ket{\psi} =
\frac{1}{\sqrt{2}}(\ket{0} \otimes \ket{1} + \ket{1} \otimes
\ket{0})$ is ultimately robust to entanglement loss. The maximal
entanglement lifetime $\widetilde{\tau}$ is a solution of equation
$(1+e^{-\gamma t})^2 = 1+e^{\gamma' t}$ and approximately equals
$\widetilde{\tau} \approx \frac{3 \ln 3}{4\gamma + 3\gamma'}$.
This shows that entanglement is more sensitive to the decoherence
rate in the memory cell (rate of the amplitude damping process).
\end{example}

\section{Non-unital channels}

\subsection{Ultimate robustness}
\label{section:nonunital-robustness}

We continue using notation $\Phi_A$ for a completely positive map
with a single Kraus operator $A$, i.e., $\Phi_A[X] = A X
A^{\dag}$. The recent result of Ref.~\cite{aubrun-2015} suggests
that if $\Phi$ is a qubit map belonging to the interior of the
cone of positivity-preserving maps, then there exist
positive-definite operators $A$ and $B$ acting on $\mathcal{H}_2$
such that the map
\begin{equation}
\Upsilon = \Phi_A \circ \Phi \circ \Phi_B
\end{equation}

\noindent is unital. This result can be viewed as a quantum
analogue of Sinkhorn's theorem~\cite{gurvits-2004}. One can always
treat map $\Upsilon$ as diagonal in the basis of Pauli operators
because appropriate unitary rotations of input and output bases
can be attributed to operators $B$ and $A$, respectively.
Alternatively, $\Phi = \Phi_{A^{-1}} \circ \Upsilon \circ
\Phi_{B^{-1}}$. The latter equation is nothing else but a
decomposition of nonunital positive qubit map $\Phi$ through some
unital map $\Upsilon$. The time-dependent version of this relation
for quantum dynamical maps takes the form
\begin{equation}
\label{decomposition} \Phi(t) = \Phi_{A^{-1}(t)} \circ \Upsilon(t)
\circ \Phi_{B^{-1}(t)}.
\end{equation}

\begin{proposition}
\label{proposition:nonunital} Suppose a local two-qubit noise
$\Phi(t) \otimes \Phi'(t)$, where both $\Phi(t)$ and $\Phi'(t)$
adopt decompositions \eqref{decomposition} with nondegenerate
operators $A(t)$, $B(t)$, $A'(t)$, $B'(t)$ and unital diagonal
maps $\Upsilon(t)$ and $\Upsilon'(t)$. Then $\Phi(t) \otimes
\Phi'(t)$ is entanglement annihilating if and only if $\Upsilon(t)
\otimes \Upsilon'(t)$ is entanglement annihilating. Ultimate
robustness to loss of entanglement exhibits the state of the form
\begin{equation}
\label{ultimate-nonunital} \ket{\psi_{\Phi\otimes\Phi'}} =
\frac{B(\widetilde{\tau}) \otimes B'(\widetilde{\tau})
\ket{\psi_{\Upsilon\otimes\Upsilon'}}}{\sqrt{\bra{\psi_{\Upsilon\otimes\Upsilon'}}
B^{\dag}(\widetilde{\tau})B(\widetilde{\tau}) \otimes
B'(\widetilde{\tau})^{\dag}B'(\widetilde{\tau})
\ket{\psi_{\Upsilon\otimes\Upsilon'}}}},
\end{equation}

\noindent where $\ket{\psi_{\Upsilon\otimes\Upsilon'}}$ is given
by Proposition~\ref{proposition:ultimate-unital} and
$\widetilde{\tau}$ is the maximal entanglement lifetime under
noise $\Upsilon(t) \otimes \Upsilon'(t)$.
\end{proposition}

\begin{proof}
Since $\Phi(t) \otimes \Phi'(t) [\ket{\psi}\bra{\psi}] = A^{-1}(t)
\otimes A'^{-1}(t) \big( \Upsilon(t) \otimes \Upsilon'(t) \big)[
B^{-1}(t) \otimes B'^{-1}(t) \ket{\psi} \bra{\psi} B^{\dag -1}(t)
\otimes B'^{\dag -1}(t) ] A^{\dag -1}(t) \otimes A'^{\dag -1}(t)$
and both $A(t)$ and $A'(t)$ are non-degenerate, then $\Phi(t)
\otimes \Phi'(t) [\ket{\psi}\bra{\psi}]$ is separable if and only
if $\big( \Upsilon(t) \otimes \Upsilon'(t) \big)[ B^{-1}(t)
\otimes B'^{-1}(t) \ket{\psi} \bra{\psi} B^{\dag -1}(t) \otimes
B'^{\dag -1}(t) ]$ belongs to a cone of separable operators. Thus
$\Phi(t) \otimes \Phi'(t) [\ket{\psi}\bra{\psi}]$ is separable for
all $\ket{\psi}$ if and only if $\big( \Upsilon(t) \otimes
\Upsilon'(t) \big)[ B^{-1}(t) \otimes B'^{-1}(t) \ket{\psi}
\bra{\psi} B^{\dag -1}(t) \otimes B'^{\dag -1}(t) ]$ is a
separable operator for all $\ket{\psi}$. As both $B(t)$ and
$B'(t)$ are nondegenerate, the linear span of operators $B^{-1}(t)
\otimes B'^{-1}(t) \ket{\psi} \bra{\psi} B^{\dag -1}(t) \otimes
B'^{\dag -1}(t)$ for all $\ket{\psi}$ is a cone of positive
operators. Thus, $\Phi(t) \otimes \Phi'(t) [\ket{\psi}\bra{\psi}]$
is separable for all $\ket{\psi}$ if and only if $\big(
\Upsilon(t) \otimes \Upsilon'(t) \big)[ \varrho ]$ is separable
for all density operators $\varrho$, i.e., $\Upsilon(t) \otimes
\Upsilon'(t)$ is entanglement annihilating. Since $\Phi(t) \otimes
\Phi'(t) [\ket{\psi_{\Phi \otimes \Phi'}}\bra{\psi_{\Phi \otimes
\Phi'}}] \propto A^{-1}(t) \otimes A'^{-1}(t) \big( \Upsilon(t)
\otimes \Upsilon'(t) \big)[ \ket{\psi_{\Upsilon \otimes
\Upsilon'}} \bra{\psi_{\Upsilon \otimes \Upsilon'}} ] A^{\dag
-1}(t) \otimes A'^{\dag -1}(t)$, then $\Phi(t) \otimes \Phi'(t)
[\ket{\psi_{\Phi \otimes \Phi'}}\bra{\psi_{\Phi \otimes \Phi'}}]$
is entangled if and only if the state $\Upsilon(t) \otimes
\Upsilon'(t) [ \ket{\psi_{\Upsilon \otimes \Upsilon'}}
\bra{\psi_{\Upsilon \otimes \Upsilon'}} ]$ is entangled.
Therefore, \eqref{ultimate-nonunital} exhibits ultimate robustness
to loss of entanglement if $\ket{\psi_{\Upsilon\otimes\Upsilon'}}$
is ultimately robust to loss of entanglement due to unital noises
$\Upsilon(t) \otimes \Upsilon'(t)$.
\end{proof}

\subsection{Explicit decomposition of nonunital qubit maps}
\label{section:nonunital-decomposition}

To utilize Proposition~\ref{proposition:nonunital} for particular
physical systems, one needs to know explicitly the operators $A$
and $B$ as well as the unital map $\Upsilon$ in formula
\eqref{decomposition} for a given qubit channel $\Phi$. In what
follows, we develop ideas of Ref.~\cite{aubrun-2015} to find such
explicit expressions.

By a suitable choice of input and output bases, one can reduce the
matrix representation of any nonunital qubit channel $\Phi$ to the
following form~\cite{ruskai-2002}:
\begin{equation}
\label{matrix-nonunital}
M(\Phi) = \left(%
\begin{array}{cccc}
  1 & 0 & 0 & 0 \\
  t_1 & \lambda_1 & 0 & 0 \\
  t_2 & 0 & \lambda_2 & 0 \\
  t_3 & 0 & 0 & \lambda_3 \\
\end{array}%
\right).
\end{equation}

Formula $\varrho = \frac{1}{2}(I + \sum_{j=1}^3 r_j \sigma_j)$
establishes a one-to-one correspondence between qubit density
operators $\varrho$ and real Bloch vectors ${\bf r} =
(r_1,r_2,r_3)^{\top}$ satisfying $|{\bf r}| = \sqrt{\sum_{j=1}^3
r_j^2} \leqslant 1$. The Bloch vector of the density operator
$\Phi[\varrho]$ is $(\lambda_1 r_1 + t_1, \lambda_2 r_2 + t_2,
\lambda_3 r_3 + t_3)^{\top}$. From this geometrical picture, it is
not hard to see that positivity of the map $\Phi$ implies
\begin{equation}
\label{positivity} \sum_{j=1}^3 \frac{t_j^2}{(1-|\lambda_j|)^2}
\leqslant 1.
\end{equation}

\noindent This necessary condition for positivity of $\Phi$ means
that the vector ${\bf t}$ has to belong to an ellipsoid with
principal axes of length $2(1-|\lambda_j|)$, $j=1,2,3$. If
$|\lambda_j| = 1$, then $t_j = 0$, and the ratio
$\frac{t_j^2}{(1-|\lambda_j|)^2}$ should be treated as zero.

Following Ref.~\cite{aubrun-2015}, we introduce operators
$\widetilde{A} = \sqrt{S}$ and $\widetilde{B} =
(\Phi^{\dag}[S])^{-1/2}$, where the positive Hermitian operator
$S$ is a fixed point of the map $F[S] =
(\Phi[(\Phi^{\dag}[S])^{-1}])^{-1}$ and $\Phi^{\dag}$ is a dual
linear map such that ${\rm tr} \big[ \Phi^{\dag}[X] Y \big] = {\rm
tr} \big[ X \Phi[Y] \big]$ for all $X,Y$. Ref.~\cite{aubrun-2015}
shows that the map $\Phi_{\widetilde{A}} \circ \Phi \circ
\Phi_{\widetilde{B}}$ is unital and positive if $\Phi$ belongs to
the interior of the cone of positivity-preserving maps, although
matrix representation of the map $\Phi_{\widetilde{A}} \circ \Phi
\circ \Phi_{\widetilde{B}}$ is not necessarily diagonal.

\begin{figure}
\includegraphics[width=8.5cm]{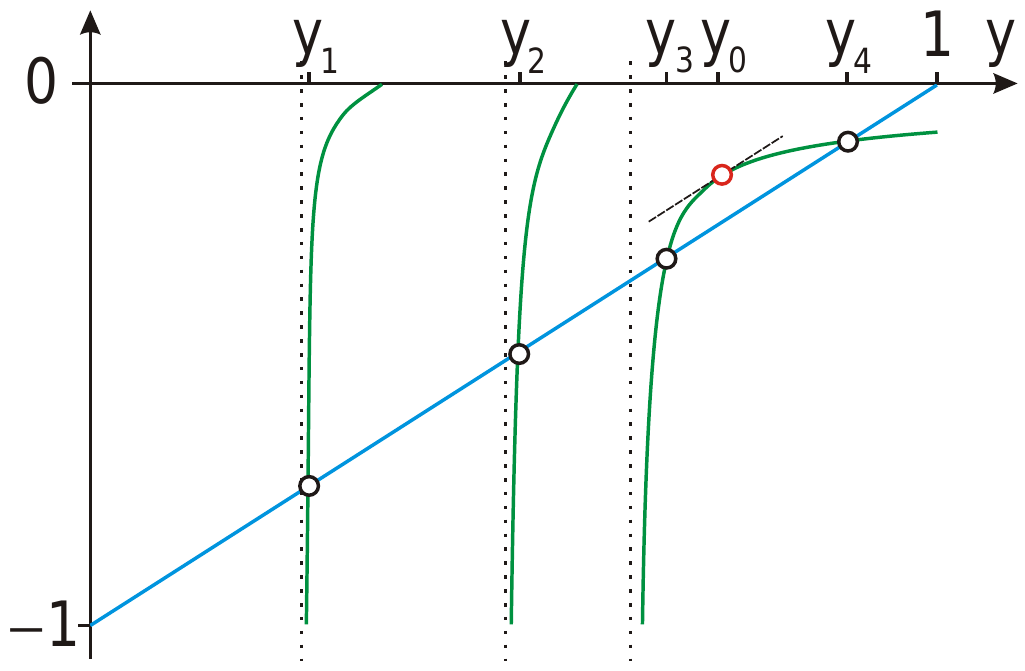}
\caption{\label{figure2} Graphical solution of Eq.~\eqref{y}, all
quantities are dimensionless.}
\end{figure}

We develop results of Ref.~\cite{aubrun-2015} and find $S$
explicitly. We fix ${\rm tr}[S] = 2$, denote $S = I + \sum_{j=1}^3
x_j \sigma_j$, and introduce a new variable $y = 1 + \sum_{j=1}^3
t_j x_j$. Then equation $S = F[S]$ reduces to
\begin{eqnarray}
&& \label{y} y - 1 = y \sum_{j=1}^3 \frac{t_j^2}{\lambda_j^2 - y}, \\
&& \label{x} x_j = \frac{y t_j}{\lambda_j^2 - y}, \quad j=1,2,3.
\end{eqnarray}

Eq.~\eqref{y} is nothing else but a quartic equation
\begin{equation}
\label{quartic} y^4 + b y^3 + c y^2 + d y + e = 0
\end{equation}

\noindent with coefficients
\begin{eqnarray}
b & = & t_1^2 + t_2^2 + t_3^2 - \lambda_1^2 - \lambda_2^2 -
\lambda_3^2 - 1, \\
c & = & \lambda_1^2 (1 - t_2^2 - t_3^2) + \lambda_2^2 (1 - t_1^2 -
t_3^2)+ \lambda_3^2 (1 - t_1^2 - t_2^2) \nonumber\\
&& + \lambda_1^2\lambda_2^2 + \lambda_2^2\lambda_3^2 +
\lambda_3^2\lambda_1^2, \\
d & = & t_1^2 \lambda_2^2 \lambda_3^2 + \lambda_1^2 t_2^2
\lambda_3^2 + \lambda_1^2 \lambda_2^2 t_3^2 - \lambda_1^2
\lambda_2^2 \lambda_3^2 \nonumber\\
&& - \lambda_1^2\lambda_2^2 - \lambda_2^2\lambda_3^2 -
\lambda_3^2\lambda_1^2, \\
e & = & \lambda_1^2 \lambda_2^2 \lambda_3^2,
\end{eqnarray}

\noindent so it can be readily solved analytically, e.g. by
Ferrari's method~\cite{kurosh-1980}. Let us demonstrate that if
$\Phi$ is a positive map, then the obtained equation has four real
non-negative roots (possibly degenerate), with the greatest one
guaranteeing positivity of operator $S$.

In fact, a graph of the left-hand side of Eq.~\eqref{y} is a line,
and a graph of the right-hand side of Eq.~\eqref{y} has (in
general) three vertical asymptotes at points $y = \lambda_j^2
\leqslant 1$; see Fig.~\ref{figure2}. Thus, Eq.~\eqref{y}
definitely has two real roots $y_{1,2} \in [0,\max \lambda_j^2)$.
The derivative of the right-hand side of Eq.~\eqref{y} equals 1 at
point $y_0
> \max \lambda_j^2$ satisfying $\sum_{j=1}^3 \frac{t_j^2 \lambda_j^2}{(\lambda_j^2 - y_0)^2} =
1$. From this follows that $y_0 \leqslant |\lambda_j|$ for all
$j=1,2,3$ because otherwise we encounter a contradiction $1 =
\sum_{j=1}^3 \frac{t_j^2 \lambda_j^2}{(\lambda_j^2 - y_0)^2} <
\sum_{j=1}^3 \frac{t_j^2}{(1-|\lambda_j|)^2} \leqslant 1$, cf.
Eq.~\eqref{positivity}. Thus, $y_0 \leqslant |\lambda_j|$ and the
right hand side of Eq.~\eqref{y} equals $y_0 - \sum_{j=1}^3
\frac{t_j^2 y_0^2}{(\lambda_j^2 - y_0)^2} \geqslant y_0 -
\sum_{j=1}^3 \frac{t_j^2 \lambda_j^2}{(\lambda_j^2 - y_0)^2} = y_0
- 1$. Therefore, at point $y=y_0$ the right-hand side of
Eq.~\eqref{y} is larger than or equal to the left-hand side of
Eq.~\eqref{y}, so Eq.~\eqref{y} has two more real roots $y_{3,4}
\in (\max \lambda_j^2,1]$; see Fig.~\ref{figure2}. Moreover, the
derivative of the right-hand side of Eq.~\eqref{y} at the largest
root $y_4$ is less than or equal to 1, which readily implies that
values $x_1,x_2,x_3$ corresponding to this root satisfy
$\sum_{j=1}^3 x_j^2 \leqslant 1$, i.e., the operator $S$ is
positive semidefinite. If $\Phi$ belongs to the interior of
positive maps, then $S$ is positive.

Calculating $\widetilde{A}$, $\widetilde{B}$ and simplifying
unitary map $\Phi_{\widetilde{A}} \circ \Phi \circ
\Phi_{\widetilde{B}}$ as much as possible, we obtain the following
result.

\begin{proposition}
\label{proposition:decomposition} Suppose a nonunital qubit map
$\Phi$, which belongs to the interior of the cone of positivity
preserving maps and is defined by matrix
representation~\eqref{matrix-nonunital}. Let the largest real root
$y$ of quartic equation~\eqref{quartic} define coefficients $x_j$,
$j=1,2,3$ by Eq.~\eqref{x}. Let $x = \sqrt{\sum_{j=1}^3 x_j^2}$
and $\xi = \sqrt{\sum_{j=1}^3 \lambda_j^2 x_j^2}$, then operators
\begin{eqnarray}
\widetilde{A} &=& \frac{\sqrt{1+x} + \sqrt{1-x}}{2} I +
\frac{\sqrt{1+x} - \sqrt{1-x}}{2x}
\sum_{j=1}^3 x_j \sigma_j, \nonumber\\
\\
\widetilde{B} &=& \frac{\sqrt{y+\xi} +
\sqrt{y-\xi}}{2\sqrt{y^2-\xi^2}} I - \frac{\sqrt{y+\xi} -
\sqrt{y-\xi}}{2 \xi \sqrt{y^2-\xi^2}}
\sum_{j=1}^3 \lambda_j x_j \sigma_j \nonumber\\
\end{eqnarray}

\noindent are Hermitian and positive; the map
$\Phi_{\widetilde{A}} \circ \Phi \circ \Phi_{\widetilde{B}}$ is
unital, positive, trace preserving, and its matrix representation
reads $M_{00}(\Phi_{\widetilde{A}} \circ \Phi \circ
\Phi_{\widetilde{B}}) = 1$, $M_{0i}(\Phi_{\widetilde{A}} \circ
\Phi \circ \Phi_{\widetilde{B}}) = M_{i0}(\Phi_{\widetilde{A}}
\circ \Phi \circ \Phi_{\widetilde{B}}) = 0$,
\begin{eqnarray}
\label{matrix-unital-non-diagonal} && M_{ij}(\Phi_{\widetilde{A}}
\circ \Phi \circ \Phi_{\widetilde{B}})
 = \frac{1-x^2}{\sqrt{y^2-\xi^2}} \Bigg[ \frac{\lambda_i
\delta_{ij}}{\sqrt{1-x^2}} \,  \nonumber\\
&& + \left(\frac{1-\sqrt{1-x^2}}{x^2
\sqrt{y^2-\xi^2}}-\frac{(y-\sqrt{y^2-\xi^2}) \lambda_i^2}{\xi ^2
\sqrt{1-x^2} y}\right) x_i \lambda_j x_j \Bigg], \qquad
\end{eqnarray}

\noindent where $i,j=1,2,3$, and $\delta_{ij}$ is the Kronecker
delta. Conventional decomposition of matrix
\eqref{matrix-unital-non-diagonal}
\begin{equation}
\label{diagonalization} \left[ M_{ij}(\Phi_{\widetilde{A}} \circ
\Phi \circ \Phi_{\widetilde{B}}) \right]_{i,j=1,2,3} =
Q_{\widetilde{U}} {\rm
diag}(\widetilde{\lambda}_1,\widetilde{\lambda}_2,\widetilde{\lambda}_3)
Q_{\widetilde{V}}
\end{equation}

\noindent with orthogonal matrices $Q_{\widetilde{U}}$ and
$Q_{\widetilde{V}}$, $\det Q_{\widetilde{U}} = \det
Q_{\widetilde{V}} = 1$, leads to the unital map $\Upsilon =
\Phi_{\widetilde{U}^{\dag}\widetilde{A}} \circ \Phi \circ
\Phi_{\widetilde{B}\widetilde{V}^{\dag}}$ with diagonal matrix
representation $M(\Upsilon) = {\rm diag}(1,
\widetilde{\lambda}_1,\widetilde{\lambda}_2,\widetilde{\lambda}_3)$.
Operators $A = \widetilde{U}^{\dag}\widetilde{A}$ and $B =
\widetilde{B}\widetilde{V}^{\dag}$.
\end{proposition}

Proposition~\ref{proposition:decomposition} allows one to reduce
any nonboundary qubit channel $\Phi$ to a unital map $\Upsilon$
with diagonal matrix representation.

The obtained result becomes particularly simple in the case $t_1 =
t_2 = 0$ because, in this case, Eq.~\eqref{y} is readily solved
and matrix~\eqref{matrix-unital-non-diagonal} is automatically
diagonal. Thus, no diagonalization~\eqref{diagonalization} is
needed, $A = \widetilde{A}$ and $B = \widetilde{B}$.

\begin{corollary}
\label{corollary-2} Suppose a nonboundary qubit channel $\Phi$
given by matrix representation~\eqref{matrix-nonunital} with $t_1
= t_2 = 0$. If
\begin{eqnarray}
\label{A-case} A &=&
\frac{2}{\sqrt{(1+t_3)^2-\lambda_3^2}+\sqrt{(1-t_3)^2-\lambda_3^2}}
\nonumber\\
&& \times \left(%
\begin{array}{cc}
  \sqrt{(1+|t_3|)^2-\lambda_3^2} & 0 \\
  0 & \sqrt{(1-|t_3|)^2-\lambda_3^2} \\
\end{array}%
\right), \\
\label{B-case} B &=& \left(%
\begin{array}{cc}
  \frac{1}{\sqrt{1+t_3 x_3 +|\lambda_3 x_3|}} & 0 \\
  0 & \frac{1}{\sqrt{1+t_3 x_3 -|\lambda_3 x_3|}} \\
\end{array}%
\right),\\
x_3 & = & - t_3
\frac{1-t_3^2+\lambda_3^2+\sqrt{[(1+t_3)^2-\lambda_3^2][(1-t_3)^2-\lambda_3^2]}}{1-t_3^2-\lambda_3^2+\sqrt{[(1+t_3)^2-\lambda_3^2][(1-t_3)^2-\lambda_3^2]}},
\nonumber
\end{eqnarray}

\noindent then $\Upsilon = \Phi_A \circ \Phi \circ \Phi_B$ is a
unital qubit channel with eigenvalues
\begin{eqnarray}
\label{lambda-tilde-1} \widetilde{\lambda}_1 & = &
\frac{2\lambda_1}{\sqrt{(1+\lambda_3)^2-t_3^2}+\sqrt{(1-\lambda_3)^2-t_3^2}},\\
\widetilde{\lambda}_2 & = &
\frac{2\lambda_2}{\sqrt{(1+\lambda_3)^2-t_3^2}+\sqrt{(1-\lambda_3)^2-t_3^2}},\\
\label{lambda-tilde-3} \widetilde{\lambda}_3 & = &
\frac{4\lambda_3}{\left(\sqrt{(1+\lambda_3)^2-t_3^2}+\sqrt{(1-\lambda_3)^2-t_3^2}\right)^2}.
\end{eqnarray}
\end{corollary}

\subsection{Generalized amplitude damping processes at finite temperature}
\label{section:nonunital-gad}

A two-level system with energy level separation $\Delta E$ is
coupled with a reservoir of finite temperature $T$, which results
in a generalized amplitude damping process
\begin{eqnarray}
\frac{d\varrho}{dt} &=& \gamma w \left( 2 \sigma_{+} \varrho
\sigma_{-} - \left\{ \sigma_{-}\sigma_{+} , \varrho
\right\} \right) \nonumber\\
&& + \gamma (1-w) \left( 2 \sigma_{-} \varrho \sigma_{+} - \left\{
\sigma_{+}\sigma_{-} , \varrho \right\} \right),
\end{eqnarray}

\noindent where $w,1-w$ are the populations of ground and excited
levels in thermal equilibrium, i.e., $\frac{1-w}{w} =
\exp(-\frac{\Delta E}{kT})$. The resulting dynamical map $\Phi(t)$
is nonunital, and its matrix representation is
\begin{equation}
M(\Phi(t)) = \left(%
\begin{array}{cccc}
  1 & 0 & 0 & 0 \\
  0 & e^{-\gamma t} & 0 & 0 \\
  0 & 0 & e^{-\gamma t} & 0 \\
  (2w-1)(1-e^{-2\gamma t}) & 0 & 0 & e^{- 2 \gamma t} \\
\end{array}%
\right).
\end{equation}

Using Corollary~\ref{corollary-2}, we find the corresponding
unital dynamical map $\Upsilon(t)$ with eigenvalues
\begin{eqnarray}
&& \label{eigenvalues-reduced-gad} \widetilde{\lambda}_1(t) =
\widetilde{\lambda}_2(t)  =
 e^{-\gamma t}  \Big\{ \sqrt{w(1-w)} (1-e^{-2\gamma t}) \nonumber\\
 && \qquad \quad + \sqrt{[1-w(1-e^{-2\gamma t})][w+e^{-2\gamma t}(1-w)]} \Big\}^{-1} \qquad
\end{eqnarray}

\noindent and $\widetilde{\lambda}_3(t) =
\widetilde{\lambda}_1^2(t) = \widetilde{\lambda}_2^2(t)$. The
latter relation means that $\Upsilon(t)$ is nothing else but an
amplitude damping process with the infinite temperature of the
environment considered in Examples~\ref{example-1} and
\ref{example-2}, although the generator of $\Upsilon(t)$ is
time-dependent due to a time deformation. Exploiting
Eq.~\eqref{B-case}, we also find
\begin{eqnarray}
\label{B-gad}
B(t) & \propto & \sqrt[4]{(1-w)[1 - (1-w)(1-e^{-2 \gamma t})]} \, \sigma_+ \sigma_- \nonumber\\
&& + \sqrt[4]{w[1 - w(1-e^{-2 \gamma t})]} \, \sigma_- \sigma_+.
\end{eqnarray}

\begin{example}
\label{example-3} Suppose two identical qubits, each experiencing
amplitude damping in a reservoir with a finite temperature $T$
such that $w,1-w$ are the populations of ground and excited levels
in thermal equilibrium (the case of two memory
qubits~\cite{fonseca-romero-2012}). What is the optimal
preparation of the initial entangled state, whose entanglement
lifetime is the longest? Surprisingly, it is not the maximally
entangled state. Using Proposition~\ref{proposition:nonunital} and
Eq.~\eqref{B-gad}, we conclude that ultimate robustness exhibits
the state
\begin{eqnarray}
&& \label{ultimately-robust-gad} \ket{\psi_{\Phi \otimes \Phi}} = \sqrt{\frac{(1-w)[1 - (1-w)(1-e^{-2 \gamma \widetilde{\tau}})]}{1-(1-2w+2w^2)(1-e^{-2 \gamma \widetilde{\tau}})}} \, \ket{0} \otimes \ket{1} \nonumber\\
&& + \sqrt{\frac{w[1 - w(1-e^{-2 \gamma
\widetilde{\tau}})]}{1-(1-2w+2w^2)(1-e^{-2 \gamma
\widetilde{\tau}})}} \, \ket{1} \otimes \ket{0},
\end{eqnarray}

\noindent where $\widetilde{\tau}$ is the maximal entanglement
lifetime under unital noise $\Upsilon(t) \otimes \Upsilon(t)$.
Using Corollary~\ref{corollary} and the explicit form of
eigenvalues \eqref{eigenvalues-reduced-gad}, we get
\begin{equation}
\widetilde{\tau} = \frac{1}{2\gamma} \ln
\frac{4(\sqrt{2}+1)w(1-w)}{1 \! + \! 4(\sqrt{2} \! + \! 1) w (1 \!
- \! w) \! - \!\sqrt{1 \! + \! 8 (\sqrt{2} \! + \! 1) w (1 \! - \!
w)}},
\end{equation}

\noindent which is much greater than the entanglement lifetime of
the maximally entangled state $\ket{\psi_+}$,
\begin{equation}
\tau_{\psi_+} = \frac{1}{2\gamma} \ln
\frac{1+\sqrt{2w(1-w)}}{\sqrt{2w(1-w)}}.
\end{equation}

\begin{figure}
\includegraphics[width=8.5cm]{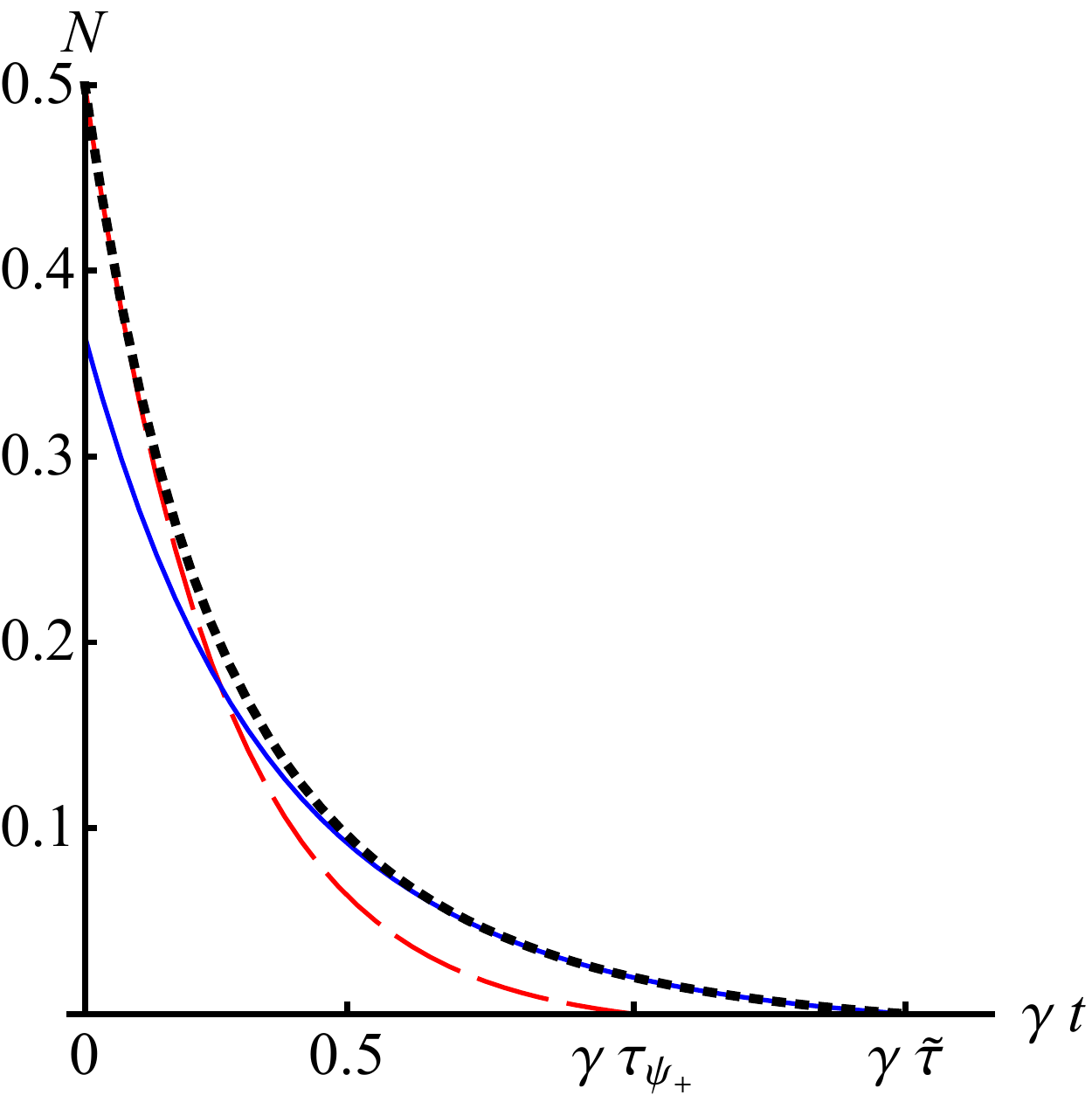}
\caption{\label{figure3} (Color online) Evolution of negativity
under local generalized amplitude damping noise $\Phi(t) \otimes
\Phi(t)$ with $w=0.01$ for the following initial states: the
maximally entangled state (red dashed curve), the ultimately
robust state (blue solid curve). $\gamma t$ is dimensionless time.
The dotted line represents a collection of negativities for states
$\Phi(t) \otimes \Phi(t)[\ket{\psi_{t}}\bra{\psi_{t}}]$, where the
interpolating initial state $\ket{\psi_{t}}$ is given by
Eq.~\eqref{interpolation-robust-gad}.}
\end{figure}

If $w \rightarrow 0$, then $\widetilde{\tau} / \tau_{\psi_+}
\rightarrow 2$, i.e., the use of the ultimately robust state
allows one to prolong the entanglement lifetime twice as compared
with the entanglement lifetime of the maximally entangled state. A
comparison of entanglement dynamics for initial states
$\ket{\psi_{\Phi \otimes \Phi}}$ and $\ket{\psi_+}$ is depicted in
Fig.~\ref{figure3}. We use negativity $N(\varrho) = \frac{1}{2}
\left( \|\varrho^{\Gamma}\|_1-1 \right)$ as the entanglement
measure of the state $\varrho$~\cite{zyczkowski-1998,vidal-2002}
($\varrho^{\Gamma}$ is the patial transpose of $\varrho$ with
respect to one of the qubits).

Finally, $\Phi(t)\otimes\Phi(t)$ is entanglement annihilating if
and only if
\begin{equation}
1-e^{-2\gamma t} \geqslant
\frac{\sqrt{1+8(\sqrt{2}+1)w(1-w)}-1}{4(\sqrt{2}+1)w(1-w)}.
\end{equation}

\noindent This result solves the problem of characterizing
entanglement annihilation by generalized amplitude damping noises
raised in Ref.~\cite{filippov-rybar-ziman-2012}.
\end{example}

Although the state~\eqref{ultimately-robust-gad} is less entangled
initially, it remains entangled longer than the maximally
entangled state $\ket{\psi_+}$, whose entanglement is greater in
the beginning of evolution; see Fig.~\ref{figure3}. Thus, the
state~\eqref{ultimately-robust-gad} is optimal for preserving
entanglement as long as possible, whereas the maximally entangled
state $\ket{\psi_+}$ is optimal for a short storage of
entanglement. In practice, however, one may be interested in
storing entanglement for some intermediate time $t_0$. An
interpolation between $\ket{\psi_+}$ and the
state~\eqref{ultimately-robust-gad} is the normalized state
\begin{eqnarray}
\label{interpolation-robust-gad} \ket{\psi_{t_0}} & \!\! \propto \!\! & \left\{ (1 \! - \! w)[1 - (1 \! - \! w)(1 \! - \! e^{-2 \gamma \widetilde{\tau}})] \right\}^{\frac{t_0}{2\widetilde{\tau}}} \ket{0} \otimes \ket{1} \nonumber\\
&& + \left\{ w[1 - w(1-e^{-2 \gamma \widetilde{\tau}})]
\right\}^{\frac{t_0}{2\widetilde{\tau}}} \ket{1} \otimes \ket{0}.
\end{eqnarray}

\noindent  One can see that the state
$\Phi(t)\otimes\Phi(t)[\ket{\psi_{t_0}}\bra{\psi_{t_0}}]$ has a
high degree of entanglement at time moment $t_0$, which is
illustrated by negativity in Fig.~\ref{figure3}. Thus, using the
state \eqref{interpolation-robust-gad} as the initial state, one
is able to reach a high degree of entanglement at time $t_0$.

In general, if a large degree of entanglement is desired at time
$t_0$, then the interpolation for the optimally prepared state is
a modification of Eq.~\eqref{ultimate-nonunital},
\begin{equation}
\label{general-robust} \ket{\psi_{\Phi\otimes\Phi'}(t_0)} \propto
[B(\widetilde{\tau}) \otimes
B'(\widetilde{\tau})]^{t_0/\widetilde{\tau}}
\ket{\psi_{\Upsilon\otimes\Upsilon'}}.
\end{equation}

The state \eqref{general-robust} always differs from the maximally
entangled state $\ket{\psi_+}$ if at least one of the noises
$\Phi(t)$ and $\Phi'(t)$ is nonunital and $t_0 > 0$.

\begin{example}
\label{example-4} Suppose a pair of entangled qubits, with the
first qubit experiencing generalized amplitude damping in a memory
cell (parameters $w$, $\gamma$) and the second (itinerant) qubit
being affected by a depolarizing noise with rate $\gamma'$.
Suppose it takes time $t_0$ for the second qubit to reach another
laboratory, after which an experiment with two apart qubits is
performed. Maximal entanglement lifetime $\widetilde{\tau}$ is a
solution of equation $(1+\widetilde{\lambda}_1(t))^2 = 1 +
e^{\gamma' t}$, where $\widetilde{\lambda}_1(t)$ is given by
Eq.~\eqref{eigenvalues-reduced-gad}. Since operator $B$ is defined
by Eq.~\eqref{B-gad} and operator $B'=I$ in this case, then the
optimal initial state guaranteeing a high degree of final
entanglement for $t_0 \in [0,\widetilde{\tau})$ is
\begin{eqnarray}
\ket{\psi_{t_0}} & \!\! \propto \!\! & \left\{ (1 \! - \! w)[1 - (1 \! - \! w)(1 \! - \! e^{-2 \gamma \widetilde{\tau}})] \right\}^{\frac{t_0}{4\widetilde{\tau}}} \ket{0} \otimes \ket{1} \nonumber\\
&& + \left\{ w[1 - w(1-e^{-2 \gamma \widetilde{\tau}})]
\right\}^{\frac{t_0}{4\widetilde{\tau}}} \ket{1} \otimes \ket{0}
\end{eqnarray}

\noindent Note that this state is different from the
state~\eqref{interpolation-robust-gad}.
\end{example}

\section{Conclusions}
\label{section:conclusions}

We have analyzed entanglement dynamics of two-qubit entangled
states subjected to local qubit noises of the most general form.

If the noise is unital, then the ultimately robust state to
entanglement loss is maximally entangled. We have found a
criterion (Proposition~\ref{proposition:unital}), which allows one
to find the maximal entanglement lifetime in this case.

If the noise is nonunital, then we have reduced this problem to
the previous one by developing a decomposition technique suggested
in Ref.~\cite{aubrun-2015}. Hereby, we have solved the problem of
full characterization of local two-qubit entanglement annihilating
channels raised in Ref.~\cite{filippov-rybar-ziman-2012}.
Moreover, explicit decomposition of nonunital qubit maps
\eqref{decomposition} can find further applications in the
analysis of $n$-tensor stable positive
maps~\cite{muller-hermes-2016,filippov-magadov-2017}, absolutely
separating quantum maps~\cite{filippov-magadov-jivulescu-2017},
and evaluation of channel capacities.

The ultimately robust state turns out to differ from the maximally
entangled one for nonunital noises. By examples of generalized
amplitude damping noises, we show that the ultimately robust state
remains entangled about twice as long as compared with the
maximally entangled one if environment temperature tends to zero.
This fact shows that the use of an ultimately robust entangled
state is beneficial for entanglement preservation. The
communication length for entanglement based protocols can be
significantly increased by using optimal state preparation.
Similarly, disentanglement time in physically implementable
systems, e.g., electron spins, could be increased as compared to
the disentanglement time for maximally entangled initial
states~\cite{mazurek-2014,bragar-2015}.

Finally, we construct an interpolation initial state, which has a
high degree of entanglement for a particular time moment $t$. This
state is close to the maximally entangled state if $t$ tends to
zero and to the ultimately robust state if $t$ approaches the
maximal entanglement lifetime.

\begin{acknowledgments}
We thank David Reeb for bringing
Refs.~\cite{aubrun-2015,gurvits-2004} to our attention. We thank
the anonymous referee for very valuable comments to improve the
quality of the manuscript. The study is supported by the Russian
Science Foundation under project No. 16-11-00084 and performed at
the Moscow Institute of Physics and Technology.
\end{acknowledgments}


\begin{thebibliography}{99}

\bibitem{gisin-thew-2007}
N. Gisin and R. Thew, Quantum communication, Nature Photonics {\bf
1}, 165 (2007).

\bibitem{bennett-1992}
C. Bennett and S. Wiesner, Communication via one- and two-particle
operators on Einstein-Podolsky-Rosen states, Phys. Rev. Lett. {\bf
69}, 2881 (1992).

\bibitem{bennett-1993}
C. H. Bennett, G. Brassard, C. Cr\'{e}peau, R. Jozsa, A. Peres,
and W. K. Wootters, Teleporting an unknown quantum state via dual
classical and Einstein-Podolsky-Rosen channels, Phys. Rev. Lett.
{\bf 70}, 1895 (1993).

\bibitem{bennett-2005}
C. H. Bennett, P. Hayden, D. W. Leung, P. W. Shor, and A. Winter,
Remote preparation of quantum states, IEEE Trans. Inform. Theory
{\bf 51}, 56 (2005).

\bibitem{ekert-1991}
A. K. Ekert, Quantum Cryptography based on Bell's Theorem, Phys.
Rev. Lett. {\bf 67}, 661 (1991).

\bibitem{deutsch-1996}
D. Deutsch, A. Ekert, R. Jozsa, C. Macchiavello, S. Popescu, and
A. Sanpera, Quantum Privacy Amplification and the Security of
Quantum Cryptography over Noisy Channels, Phys. Rev. Lett. {\bf
77}, 2818 (1996).

\bibitem{acin-2007}
A. Ac\'{i}n, N. Brunner, N. Gisin, S. Massar, S. Pironio, and V.
Scarani, Device-Independent Security of Quantum Cryptography
against Collective Attacks, Phys. Rev. Lett. {\bf 98}, 230501
(2007).

\bibitem{horodecki-2009}
R. Horodecki, P. Horodecki, M. Horodecki, and K. Horodecki,
Quantum entanglement, Rev. Mod. Phys. {\bf 81}, 865 (2009).

\bibitem{werner-1989}
R. F. Werner, Quantum states with Einstein-Podolsky-Rosen
correlations admitting a hidden-variable model, Phys. Rev. A {\bf
40}, 4277 (1989).

\bibitem{nielsen-1999}
M. A. Nielsen, Conditions for a Class of Entanglement
Transformations, Phys. Rev. Lett. {\bf 83}, 436 (1999).

\bibitem{zukowski-1993}
M. \.{Z}ukowski, A. Zeilinger, M. A. Horne, and A. K. Ekert,
``Event-ready-detectors'' Bell experiment via entanglement
swapping, Phys. Rev. Lett. {\bf 71}, 4287 (1993).

\bibitem{nielsen-2000}
M. A. Nielsen and I. L. Chuang, {\it Quantum Computation and
Quantum Information} (Cambridge University Press, Cambridge,
2000).

\bibitem{breuer-petruccione-2002}
H.-P. Breuer and F. Petruccione, {\it The Theory of Open Quantum
Systems} (Oxford University Press, Oxford, 2002).

\bibitem{holevo-2012}
A. S. Holevo, {\it Quantum Systems, Channels, Information}
(Walter de Gruyter, Berlin, 2012).

\bibitem{yu-2004}
T. Yu, J.H. Eberly, Finite-Time Disentanglement Via Spontaneous
Emission, Phys. Rev. Lett. {\bf 93}, 140404 (2004).

\bibitem{bellomo-2008}
B. Bellomo, R. Lo Franco, and G. Compagno, Entanglement dynamics
of two independent qubits in environments with and without memory,
Phys. Rev. A {\bf 77}, 032342 (2008).

\bibitem{yu-eberly-2009}
T. Yu, J. H. Eberly, Sudden Death of Entanglement, Science {\bf
323}, 598 (2009).

\bibitem{wang-2013}
C. Wang and Q.-H. Chen, Exact dynamics of quantum correlations of
two qubits coupled to bosonic baths, New J. Phys. {\bf 15}, 103020
(2013).

\bibitem{shaham-2015}
A. Shaham, A. Halevy, L. Dovrat, E. Megidish, and H. S.
Eisenberga, Entanglement dynamics in the presence of controlled
unital noise, Sci. Rep. {\bf 5}, 10796 (2015).

\bibitem{aolita-2015}
L. Aolita, F. de Melo, and L. Davidovich, Open-system dynamics of
entanglement: a key issues review, Rep. Prog. Phys. {\bf 78},
042001 (2015).

\bibitem{horodecki-prl-1997}
M. Horodecki, P. Horodecki, and R. Horodecki, Inseparable Two
Spin-$\frac{1}{2}$ Density Matrices Can Be Distilled to a Singlet
Form, Phys. Rev. Lett. {\bf 78}, 574 (1997).

\bibitem{moravcikova-ziman-2010}
L. Morav\v{c}\'{i}kov\'{a} and M. Ziman, Entanglement-annihilating
and entanglement-breaking channels, J. Phys. A: Math. Theor. {\bf
43}, 275306 (2010).

\bibitem{filippov-rybar-ziman-2012}
S. N. Filippov, T. Ryb\'{a}r, and M. Ziman, Local two-qubit
entanglement-annihilating channels, Phys. Rev. A {\bf 85}, 012303
(2012).

\bibitem{filippov-ziman-2013}
S. N. Filippov and M. Ziman, Bipartite entanglement-annihilating
maps: Necessary and sufficient conditions, Phys. Rev. A {\bf 88},
032316 (2013).

\bibitem{filippov-ziman-2-2013}
S. N. Filippov, A. A. Melnikov, and M. Ziman, Dissociation and
annihilation of multipartite entanglement structure in dissipative
quantum dynamics, Phys. Rev. A {\bf 88}, 062328 (2013).

\bibitem{filippov-2014}
S. N. Filippov, PPT-Inducing, Distillation-Prohibiting, and
Entanglement-Binding Quantum Channels, J. Russ. Laser Res. {\bf
35}, 484 (2014).

\bibitem{filippov-2015}
S. N. Filippov, Influence of Deterministic Attenuation and
Amplification of Optical Signals on Entanglement and Distillation
of Gaussian and Non-Gaussian Quantum States, EPJ Web of
Conferences {\bf 103}, 03003 (2015).

\bibitem{peres-1996}
A. Peres, Separability Criterion for Density Matrices, Phys. Rev.
Lett. {\bf 77}, 1413 (1996).

\bibitem{horodecki-1996}
M. Horodecki, P. Horodecki, and R. Horodecki, Separability of
mixed states: necessary and sufficient conditions, Phys. Lett. A
{\bf 223}, 1 (1996).

\bibitem{hill-wootters-1997}
S. Hill and W. K. Wootters, Entanglement of a Pair of Quantum
Bits, Phys. Rev. Lett. {\bf 78}, 5022 (1997).

\bibitem{wootters-1998}
W. K. Wootters, Entanglement of Formation of an Arbitrary State of
Two Qubits, Phys. Rev. Lett. {\bf 80}, 2245 (1998).

\bibitem{jakobczyk-2004}
L. Jak\'{o}bczyk and A. Jamr\'{o}z, Noise-induced finite-time
disentanglement in two-atomic system, Phys. Lett. A {\bf 333}, 35
(2004).

\bibitem{ali-2009}
M. Ali, G. Alber, and A. R. P. Rau, Manipulating entanglement
sudden death of two-qubit X-states in zero- and finite-temperature
reservoirs, J. Phys. B: At. Mol. Opt. Phys. {\bf 42}, 025501
(2009).

\bibitem{horodecki-2003}
M. Horodecki, P. W. Shor, and M. B. Ruskai, Entanglement Breaking
Channels, Rev. Math. Phys. {\bf 15}, 629 (2003).

\bibitem{konrad-2008}
T. Konrad, F. de Melo, M. Tiersch, C. Kasztelan, A. Arag\~{a}o,
and A. Buchleitner, Evolution equation for quantum entanglement,
Nature Physics {\bf 4}, 99 (2008).

\bibitem{xu-2009}
J.-S. Xu, C.-F. Li, X.-Y. Xu, C.-H. Shi, X.-B. Zou, and G.-C. Guo,
Experimental Characterization of Entanglement Dynamics in Noisy
Channels, Phys. Rev. Lett. {\bf 103}, 240502 (2009).

\bibitem{tiersch-2009}
M. Tiersch, {\it Benchmarks and Statistics of Entanglement
Dynamics}, Ph.D. thesis, University of Freiburg (2009),
http://www.freidok.uni-freiburg.de/volltexte/6878/ .

\bibitem{wang-2015}
X.-W. Wang, S.-Q. Tang, J.-B. Yuan, and L.-M. Kuang, Nonmaximally
Entangled States can be Better for Quantum Correlation
Distribution and Storage, Int. J. Theor. Phys. {\bf 54}, 1461
(2015).

\bibitem{ziman-buzek-2005}
M. Ziman and V. Bu\v{z}ek, Entanglement-induced state ordering
under local operations, Phys. Rev. A {\bf 73}, 012312 (2006).

\bibitem{ziman-buzek-2007}
M. Ziman and V. Bu\v{z}ek, Entanglement Measures: State ordering
vs. Local Operations, in \textit{Quantum Communication and
Security} (edited by M. \.{Z}ukowski, S. Kilin, J. Kowalik), pp.
196-204 (IOS Press, Amsterdam, 2007).

\bibitem{landau-streater-1993}
L. J. Landau and R. F. Streater, On Birkhoff's theorem for doubly
stochastic completely positive maps of matrix algebras, Linear
Algebra Appl. {\bf 193}, 107 (1993).

\bibitem{ruskai-2002}
M. B. Ruskai, S. Szarek, and E. Werner, An analysis of
completely-positive trace-preserving maps on $M_2$, Linear Algebra
Appl. {\bf 347}, 159 (2002).

\bibitem{ruskai-2003}
M. B. Ruskai, Qubit Entanglement Breaking Channels, Rev. Math.
Phys. {\bf 15}, 643 (2003).

\bibitem{horodecki-1999}
M. Horodecki and P. Horodecki, Reduction criterion of separability
and limits for a class of distillation protocols, Phys. Rev. A
{\bf 59}, 4206 (1999).

\bibitem{aubrun-2015}
G. Aubrun and S. J. Szarek, Two proofs of St{\o}rmer's theorem
arXiv:1512.03293.

\bibitem{gurvits-2004}
L. Gurvits, Classical complexity and quantum entanglement, J.
Comput. Syst. Sci. {\bf 69}, 448 (2004).

\bibitem{kurosh-1980}
A. Kurosh, {\it Higher Algebra} (Mir Publishers, Moscow, 1980).

\bibitem{fonseca-romero-2012}
K. M. Fonseca-Romero, J. Martínez-Rinc\'{o}n, and C. Viviescas,
Statistical portrait of the entanglement decay of two-qubit
memories, Phys. Rev. A {\bf 86}, 042325 (2012).

\bibitem{zyczkowski-1998}
K. \.{Z}yczkowski, P. Horodecki, A. Sanpera, and M. Lewenstein,
Volume of the set of separable states, Phys. Rev. A {\bf 58}, 883
(1998).

\bibitem{vidal-2002}
G. Vidal and R. F. Werner, Computable measure of entanglement,
Phys. Rev. A {\bf 65}, 032314 (2002).

\bibitem{muller-hermes-2016}
A. M\"{u}ller-Hermes, D. Reeb, and M. M. Wolf, Positivity of
linear maps under tensor powers, J. Math. Phys. {\bf 57}, 015202
(2016).

\bibitem{filippov-magadov-2017}
S. N. Filippov and K. Y. Magadov, Positive tensor products of maps
and $n$-tensor-stable positive qubit maps, J. Phys. A: Math.
Theor. {\bf 50}, 055301 (2017).

\bibitem{filippov-magadov-jivulescu-2017}
S. N. Filippov, K. Y. Magadov, and M. A. Jivulescu, Absolutely
separating quantum maps and channels, New J. Phys. {\bf 19},
083010 (2017).

\bibitem{mazurek-2014}
P. Mazurek, K. Roszak, R. W. Chhajlany, and P. Horodecki,
Sensitivity of entanglement decay of quantum-dot spin qubits to
the external magnetic field, Phys. Rev. A {\bf 89}, 062318 (2014).

\bibitem{bragar-2015}
I. Bragar and \L. Cywi\'{n}ski, Dynamics of entanglement of two
electron spins interacting with nuclear spin baths in quantum
dots, Phys. Rev. B {\bf 91}, 155310 (2015).

\end{thebibliography}
\end{document}